\documentclass[reprint,aps,superscriptaddress,citeautoscript]{revtex4-1}
\usepackage{graphicx}
\usepackage{setspace,transparent}
\usepackage{color}
\usepackage{placeins}
\usepackage{float} 
\usepackage{fontenc,ragged2e,wrapfig}
\usepackage{mathtools,amssymb,amsmath,nicefrac,mathrsfs}
\usepackage{amsthm}
\usepackage[font=footnotesize]{caption}
\usepackage{soul}
\usepackage[backref=page]{hyperref}
\setcitestyle{super}

\usepackage[dvipsnames]{xcolor}
\usepackage[utf8]{inputenc}
\usepackage{listings}

\usepackage{titlesec}

\usepackage[resetlabels,labeled]{multibib}
\newcites{sec}{Reference}

\newtheorem{theorem}{Theorem}
\newtheorem{lemma}{Lemma}

\DeclareMathAlphabet{\mathpzc}{OT1}{pzc}{m}{it}
\DeclareCaptionJustification{justified}{\justifying}
\newcommand{\RNum}[1]{\uppercase\expandafter{\romannumeral #1\relax}}

\newsavebox\mysavebox

\begin{document}

\title{Unveiling Explosive Vulnerability of Networks through Edge Collective Behavior}

\author{\textsc{Peng Peng}}
\affiliation{Institute of Fundamental and Frontier Studies, University of Electronic Science and Technology of China, 611731 Chengdu, P.R. China}

\author{\textsc{Tianlong Fan}}
\email{Correspondence should be addressed to: \\ tianlong.fan@ustc.edu.cn, linyuan.lv@uestc.edu.cn}
\affiliation{School of Cyber Science and Technology, University of Science and Technology of China, 230026 Hefei, P.R. China}

\author{\textsc{Xiao-Long Ren}}
\affiliation{Yangtze Delta Region Institute (Huzhou), University of Electronic Science and Technology of China, 313001 Huzhou, P.R. China}

\author{\textsc{Linyuan L\"{u}}}
\email{Correspondence should be addressed to: \\ tianlong.fan@ustc.edu.cn, linyuan.lv@uestc.edu.cn}
\affiliation{Institute of Fundamental and Frontier Studies, University of Electronic Science and Technology of China, 611731 Chengdu, P.R. China}
\affiliation{School of Cyber Science and Technology, University of Science and Technology of China, 230026 Hefei, P.R. China}

%106 words

\date{\today}

\maketitle
\begin{spacing}{0.91} 

\textbf{\small Edges, binding together nodes within networks, have the potential to induce dramatic transitions when specific collective failure behaviors emerge. These changes, initially unfolding covertly and then erupting abruptly, pose substantial, unforeseeable threats to networked systems, and are termed explosive vulnerability. Thus, identifying influential edges capable of triggering such drastic transitions, while minimizing cost, is of utmost significance. Here, we address this challenge by introducing edge collective influence (ECI), which builds upon the optimal percolation theory applied to line graphs. ECI embodies features of both optimal and explosive percolation, involving minimized removal costs and explosive dismantling tactic. Furthermore, we introduce two improved versions of ECI, namely IECI and IECIR, tailored for objectives of hidden and fast dismantling, respectively, with their superior performance validated in both synthetic and empirical networks. Finally, we present a dual competitive percolation (DCP) model, whose reverse process replicates the explosive dismantling process and the trajectory of the cost function of ECI, elucidating the microscopic mechanisms enabling ECI's optimization. ECI and the DCP model demonstrate the profound connection between optimal and explosive percolation. This work significantly deepens our comprehension of percolation and provides valuable insights into the explosive vulnerabilities arising from edge collective behaviors.}
% Abstract: 301-106=195words
\end{spacing}

{\small
Connectivity plays an indispensable role in the operation of networked systems~\cite{boccalettiComplex2006,newmanStructure2011} \!\!. Its transition often manifests as an abrupt eruption of severe systemic risks, triggered by the accumulation of seemingly minor perturbations that result from the failure of basic components, such as units (nodes) and connections (edges)~\cite{dsouzaExplosive2019,liPercolation2021,grossPercolation2022} \!\!. The initially imperceptible disturbances followed by a sudden eruption of severe disintegration, known as ``explosive vulnerability'', constitute two key characteristics of this phenomenon, often giving rise to significant yet challenging-to-counter systemic catastrophic risks. Hence, the identification of interactions with significant collective impact prior to the occurrence of extensive changes in the network’s structure and dynamics is of paramount importance. 

Explosive vulnerability fundamentally embodies a network dismantling paradigm. The large-scale network dismantling~\cite{braunsteinNetworkDismantling2016} \!\!, which underscores the efficient disintegration of networks by removing a minimum set of nodes or edges, falls within the class of nondeterministic polynomial hard (NP-hard) problems, presenting a persisting and substantial challenge~\cite{braunsteinNetworkDismantling2016, renGeneralizedNetworkDismantling2019} \!\!. This challenge arises in various practical scenarios, including the regular preservation of power grids to prevent large-scale cascading failures~\cite{yangSmall2017,schaferDynamically2018} \!\!, finding superblockers~\cite{radicchiFundamentalDifferenceSuperblockers2017} in epidemic networks for targeted vaccination to hinder epidemic spread~\cite{altarelliContainingEpidemicOutbreaks2014} \!\!, and disrupting criminal networks with minimal cost~\cite{renGeneralizedNetworkDismantling2019} \!\!. A considerable amount of prior research~\cite{moroneInfluenceMaximizationComplex2015,braunsteinNetworkDismantling2016,clusellaImmunizationTargetedDestruction2016,mugishaIdentifyingOptimalTargets2016,zdeborovaFastSimpleDecycling2016,tianArticulationPointsComplex2017,liuOptimizationTargetedNode2018,renGeneralizedNetworkDismantling2019,fanFindingKeyPlayers2020,grassiaMachineLearningDismantling2021} has been devoted to the dismantling problem. For example, the Collective Influence (CI) of node based on optimal percolation~\cite{moroneInfluenceMaximizationComplex2015} brings the global function of influence to our attention, the belief propagation-guided decimation (BPD)~\cite{mugishaIdentifyingOptimalTargets2016} highlights the impact of loop structures on the dismantling task, the Min-Sum algorithm~\cite{braunsteinNetworkDismantling2016} then demonstrates the collective essence of the dismantling problem, the approach of Generalized Network Dismantling (GND)~\cite{renGeneralizedNetworkDismantling2019} further takes arbitrary node removal cost into consideration, and machine learning has also been effectively applied to the dismantling task~\cite{fanFindingKeyPlayers2020, grassiaMachineLearningDismantling2021} recently. However, these advancements do not address the question of what dismantling strategies can capture the characteristic of explosive vulnerability in networks.

The explosiveness characteristic reminds us of the explosive percolation (EP)~\cite{achlioptasExplosivePercolationRandom2009}, which performs competitive addition of edges leading to sudden, profound structural shifts at a critical point~\cite{araujoExplosivePercolationControl2010,dacostaExplosivePercolationTransition2010,riordanExplosivePercolationContinuous2011,kuehnUniversalRouteExplosive2021} \!\!. It has been linked to a multitude of real-world explosive phenomena, including explosive opinion depolarization~\cite{ojerModelingExplosiveOpinion2023} \!\!, explosive synchronization~\cite{millanExplosiveHigherorderKuramoto2020,thumlerSynchronyWeakCoupling2023} \!\!, and higher-order percolation~\cite{sunHigherorderPercolationProcesses2021a,sunDynamicNaturePercolation2023a, battiston2021physics} \!\!. Additionally, the reverse percolation process has been employed to solve network dismantling problem, such as the inverse targeting immunization strategy~\cite{schneider2012inverse} and the explosive immunization (EI)~\cite{clusellaImmunizationTargetedDestruction2016} method. Since the competitive growth suppression of these methods occurs at the nodal level, each removal of a node from the giant connected component (GCC) will decrease the GCC size given the finite size of the dismantled networks, and thus the initial undetectable disturbance feature of explosive vulnerability (i.e., none decrease of the GCC size) requires the consideration of the competitive selection on edges rather than nodes. Therefore, accurately identifying the edge collective behavior exhibiting explosive vulnerability remains a pressing inquiry in need of resolution.

Optimal percolation excels in identifying nodes with the most collective influence, however, it finds relatively fewer applications at the edge level. Optimal percolation minimizes the dismantling cost by prioritizing the destruction of loops within the network~\cite{moroneInfluenceMaximizationComplex2015} \!\!. When it comes to edge removal scenarios, nodes on loops are preserved due to their redundant connections with the rest of the network. This preservation of loop nodes exhibits an element of concealment during the process of network fragmentation of breaking loops. In contrast, explosive percolation achieves early growth suppression of connectivity through global competition~\cite{naglerImpactSingleLinks2011a} \!\!, resulting in a delayed initiation of percolation phase transitions. This delay intensifies connectivity growth at the percolation threshold and requires the addition of only a minimal number of edges post-transition to complete percolation. Similarly, employing its reverse process for network dismantling enables the imperceptible removal of a few edges, leading to a sudden and explosive network disintegration. Thus, while optimal and explosive percolation may appear markedly distinct, their pronounced similarities in understanding network explosive vulnerability become evident. However, whether a deeper intrinsic connection between these optimal and explosive percolation exists remains an open question.

In this study, we capture the explosive vulnerability of networks by introducing an edge collective influence (ECI) method based on the optimal percolation theory~\cite{moroneInfluenceMaximizationComplex2015} \!\!. ECI can dismantle networks with minimal edge removal cost and in an explosive manner, combining the features of both optimal and explosive percolation. Tailoring for hidden and fast dismantling objectives, we devise two improved ECI approaches, IECI and IECIR, utilizing the reinsertion and reordering techniques respectively. Numerical experiments demonstrate that these two algorithms outperform all other competitive methods in both synthetic and empirical networks in terms of cost-effectiveness. They exhibit high performance and versatility, enabling their application across a diverse range of scenarios. Furthermore, we develop a dual competitive percolation (DCP) model whose reverse process replicates the explosive dismantling process and the trajectory of the cost function of ECI. This cost function serves as a representation of global influence. By employing competitive growth suppression at both the node and cluster levels, the DCP model unveils the pathway through which collective behavior triggers explosive vulnerability, as well as the microscopic mechanism behind the optimization of ECI and its cost function. More notably, IECI and IECIR have dismantling thresholds that can be approximated by the percolation threshold of an improved dual competitive percolation model. The dismantling threshold represents the proportion of edges removed when reaching the dismantling target.

\begin{figure}[htb]
  \centering
  \includegraphics[width=0.48\textwidth]{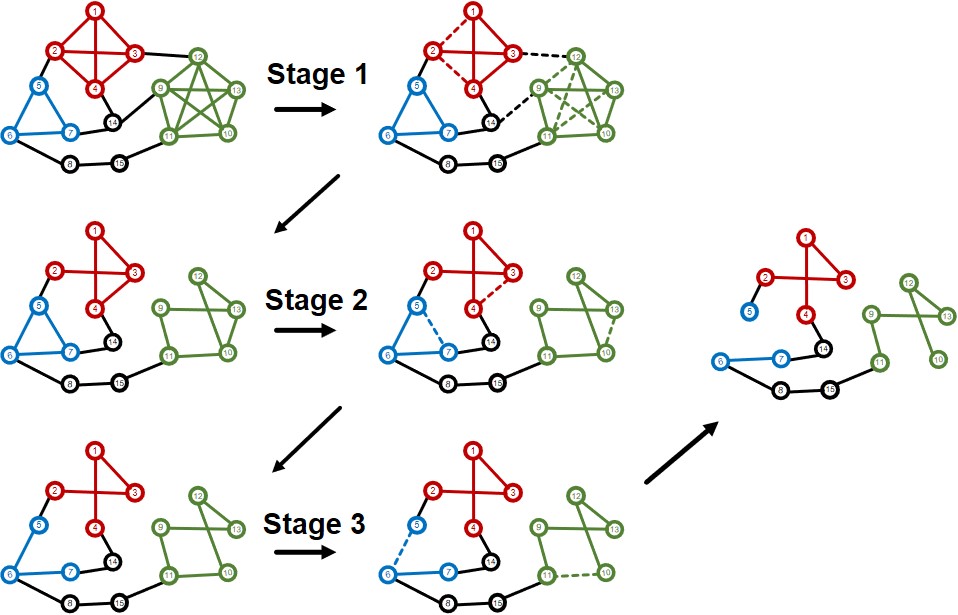}
  \caption{\textbf{The ECI dismantling process.} {\textbf{a}, Schematic of the ECI dismantling process in an example network. The dashed line shows the edge removed in each stage. In stage 1, ECI dismantles a 4-simplex (green) and a 3-simplex (red). Then in stage 2, three 2-simplices are dismantled. In stage 3, all loops are removed, resulting in a chain-like network.}
  }
  \label{fig:fig1}
  \vspace*{-0.3cm}
\end{figure}

\begin{figure*}[htb]
    \centering
    \includegraphics[width=0.8\textwidth]{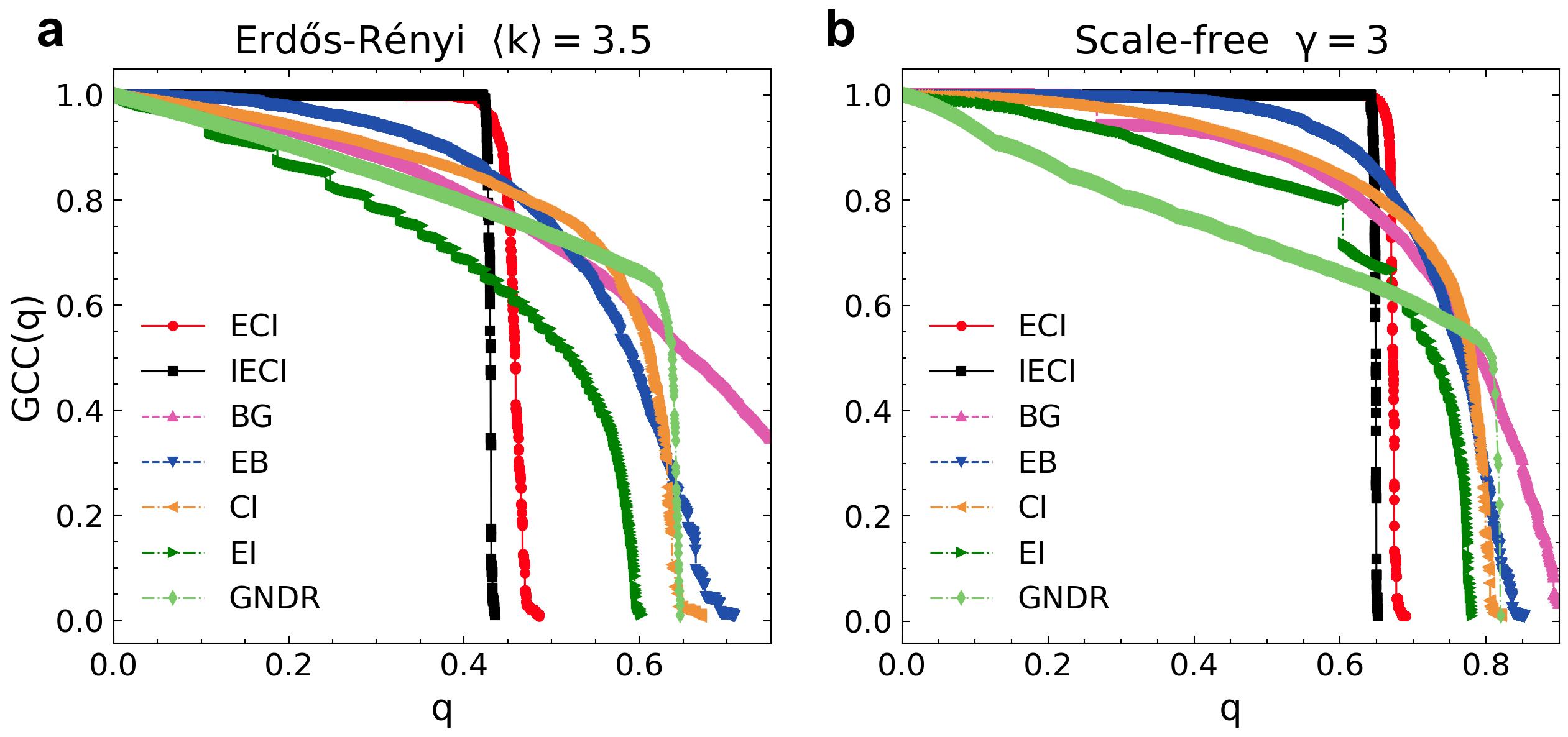}
    \caption{\textbf{Outperformance and explosive dismantling characteristic of ECI and IECI in synthetic networks.} Here, the dismantling strategies, i.e., bridgeness (BG)~\cite{chengBridgenessLocalIndex2010}, edge betweenness (EB)~\cite{girvanCommunityStructureSocial2002}, CI, EI, and generalized network dismantling with reinsertion (GNDR)~\cite{renGeneralizedNetworkDismantling2019} \!\! (see Supplementary Information Section~\ref{S3.} for definitions), are employed as comparative baselines for the END problem. In all the cases here, the dismantling target is set to $C=0.01N$. $GCC(q)$ represents the relative size of GCC after removing a fraction $q$ of edges. \textbf{a,}  $GCC(q)$ for an Erdős–Rényi network~\cite{newmanNetworks2018} (10,000 nodes and 17,500 edges, i.e., average degree $\langle k \rangle = 3.5$). 
    \textbf{b,} $GCC(q)$ for a scale-free network with degree exponent $\gamma = 3$ generated by Barabási–Albert model~\cite{barabasiEmergenceScalingRandom1999} (10,000 nodes and parameter $m=3$). }
    \label{fig:model}
    \vspace*{-0.3cm}
\end{figure*}

\begin{figure*}[htb]
    \centering
    \includegraphics[width=0.8\textwidth]{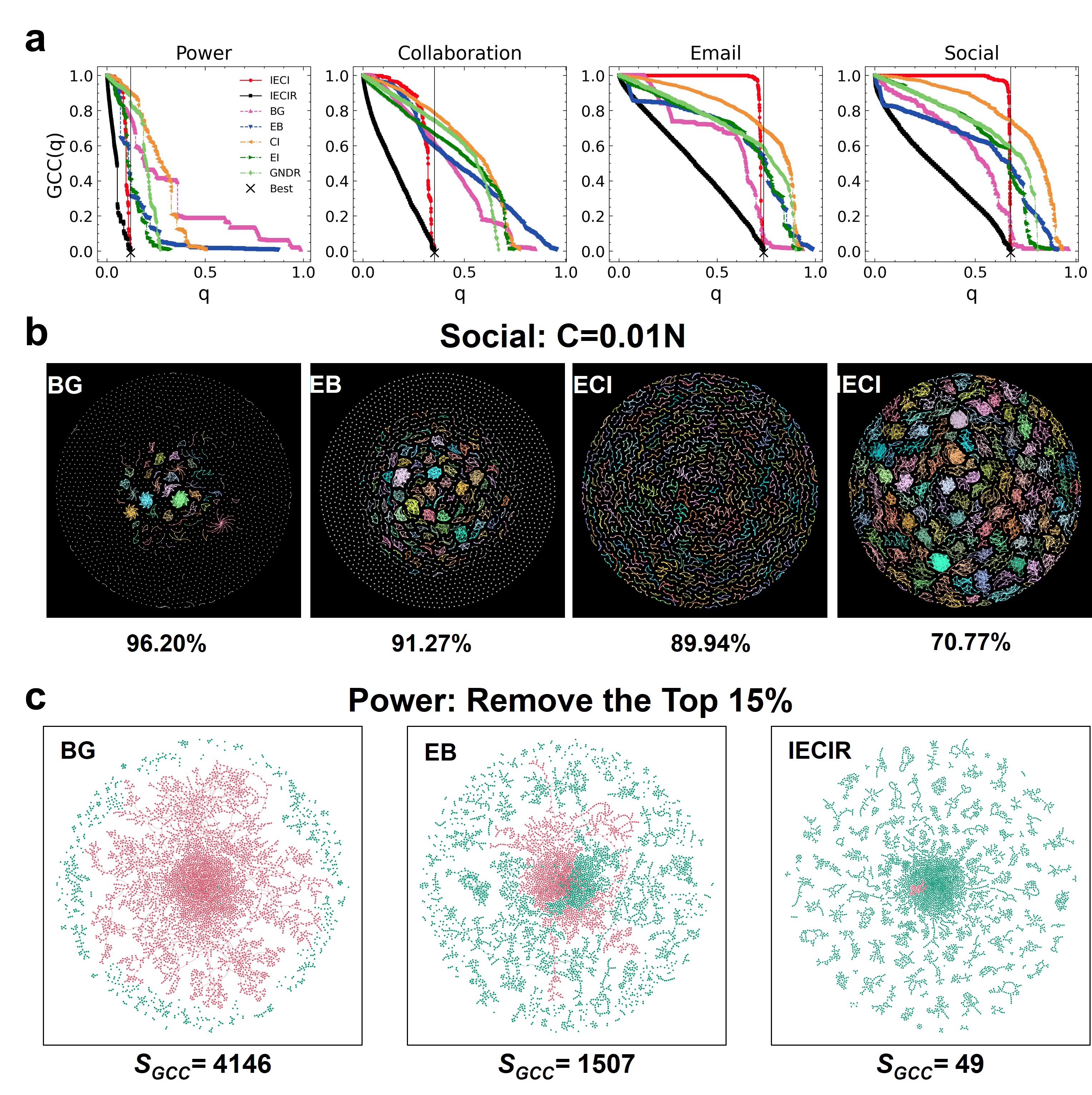}
    \caption{\textbf{Performance of IECI and IECIR in empirical networks.} We compare IECI and IECIR with some baseline dismantling methods.  \textbf{a,} The $GCC(q)$ for four empirical networks, each dismantled by IECI, IECIR, BG, EB, CI, EI, and GNDR, with the dismantling target set to $C=0.01N$. The smallest dismantling threshold among these methods is marked with black `$\times$'. The empirical networks (see Supplementary Information Section~\ref{S2.} for topology information) consist of a power-grid network~\cite{rossiNetworkDataRepository2015} (Power), a scientific collaboration network~\cite{rossiNetworkDataRepository2015} (Collaboration), an Email network~\cite{rossiNetworkDataRepository2015} (Email), and an online social network of Petster–Hamster~\cite{renGeneralizedNetworkDismantling2019} (Social). 
    \textbf{b,} A social network~\cite{renGeneralizedNetworkDismantling2019} dismantled by BG, EB, ECI, and IECI, respectively, with a dismantling target of $C=0.01N$. The percentage below each diagram represents the remove fraction $q$ of the corresponding method. \textbf{c,} A power network dismantled by BG, EB, and IECIR respectively, with a dismantling cost of the top $15\%$ edges. The $S_{GCC}$ value below each diagram represents the size of the GCC of the corresponding method, indicated by the red portion in the diagram. Note that IECIR only needs to remove $13\%$ of the total edges to achieve the dismantling target.}
    \label{fig:real}
    \vspace*{-0.3cm}
\end{figure*}

\section*{The approach to the edge-based network dismantling problem}

Given a network $G(V,E)$, where $V$ and $E$ denote the set of nodes and the set of edges in the network, respectively. $N=|V|$ and $M=|E|$ are the number of nodes and edges of the network, respectively. Then the edge-based network dismantling (END) problem can be formulated as follows: given a dismantling target $C$, find a minimum set of edges $E_{r}$ such that the size of GCC in the network $G^*(V,E\backslash E_{r})$ , denoted as $S_{GCC}$, satisfies $S_{GCC} \leq C$.

Here, we opt to generalize the optimal percolation theory~\cite{moroneInfluenceMaximizationComplex2015} to edges to identify the edge set with the most collective influence in networks. Let the vector $\mathbf{m}=\{m_1,\dots,m_M\}$ denote whether the edges in the network are removed or not. $m_i=0$ means that edge $i$ is removed and $m_i=1$ means not. The fraction of removed edges is $q=1-\frac{\sum_{i}m_i}{M}$. Let $A_m=\{a_{ij}\}$ denote the edge adjacency matrix of network $G$, in which $a_{ij}=1$ if edge $i$ is adjacent to edge $j$, and $a_{ij}=0$ otherwise. We define the generalized degree of edge $i$ as $k_i=\sum_j a_{ij}$, which denotes the number of edges adjacent to $i$.

For END, we consider the order parameter $e_{i \rightarrow j}$, i.e., the probability that edge $i$ still belongs to the GCC after removing edge $j$~\cite{newmanMessagePassingMethods2023} \!\!. For $C\rightarrow 0$, it is required that the probability that a randomly selected edge belongs to the GCC is small enough after removing some of the edges, which is consistent with $\sum_{i,j} e_{i \rightarrow j}\rightarrow0$. We consider the solution for all edges $i,j$: $\{e_{i \rightarrow j}=0\}$, and its stability requires that the largest eigenvalue $\lambda(\mathbf{m};q)$ of the modified edge-based non-backtracking operator~\cite{gloverSpectralPropertiesNonbacktracking2021} $\hat{\mathcal{M}}$ is no larger than 1, where $\hat{\mathcal{M}}=\{\mathcal{M}_{k\rightarrow l,i\rightarrow j}\}$, $\mathcal{M}_{k\rightarrow l, i\rightarrow j}\equiv \frac{\partial e_{i\rightarrow j}}{\partial e_{k\rightarrow l}}|_{\{e_{i\rightarrow j}=0\}}$~\cite{moroneInfluenceMaximizationComplex2015} \!\!. Thus, END can be solved by finding the optimal configuration $\mathbf{m}^\ast$ that minimizes $\lambda(\mathbf{m};q)$ and satisfies the dismantling target. Due to the NP-hard nature of this optimization problem, solving it with complete accuracy is not practically feasible, so we consider the substitution of $\lambda(\mathbf{m};q)$: the cost energy function $\left|\boldsymbol{w}_{\ell}(\mathbf{m})\right|^2$ of edge collective influence for a finite $\ell$ (see Supplementary Information Sections~\ref{S1.1} and~\ref{S1.2} for details), abbreviated as \textbf{cost function} in the rest of the paper. When loops are neglected, we can approximate the local environment around any edge by a tree, the leading term of $\left|\boldsymbol{w}_{\ell}(\mathbf{m})\right|^2$ can be written as
\begin{equation}
E_{\ell}(\mathbf{m})=\sum_{i=1}^{M} z_{i} \sum_{j \in \partial \operatorname{Ball}(i, \ell)}\left(\prod_{k \in \mathcal{P}_{\ell}(i, j)} m_{k}\right) z_{j},
\label{equ_1}
\end{equation}
where $z_i=k_i-1$, $\ell$ is the length of the shortest path from edge $i$ to $j$, $\operatorname{Ball}(i, \ell)$ is the set of edges $j$ inside the ball with radius $\ell$ at the center of edge $i$, $\partial \operatorname{ Ball}(i, \ell)$ is the boundary of the ball, and $\mathcal{P}_{\ell}(i, j)$ is the shortest path of length $\ell$ between $i$ and $j$. 

Our goal shifts from finding the optimal set $E_r$ to finding the optimal configuration of edges $\mathbf{m}^\ast$ that minimizes the leading term of the cost function $E_{\ell}(\mathbf{m})$. Inspired by the work of CI~\cite{moroneInfluenceMaximizationComplex2015} \!\!,
we define the edge collective influence (ECI) of edge $i$ at the level of $\ell$ as follows:
\begin{equation}
ECI_\ell (i)=z_{i} \sum_{j \in \partial \operatorname{Ball}(i, \ell)}\left(\prod_{k \in \mathcal{P}_{\ell}(i, j)} m_{k}\right) z_{j}.
\label{equ_2}
\end{equation}

When no edge is removed, $\mathbf{m}=\mathbf{1}$; then we get
\begin{equation}
ECI_\ell (i)=z_{i} \sum_{j \in \partial \operatorname{Ball}(i, \ell)}z_{j}.
\label{equ_3}
\end{equation}

\begin{figure*}[htb]
    \centering
    \includegraphics[width=0.9\textwidth]{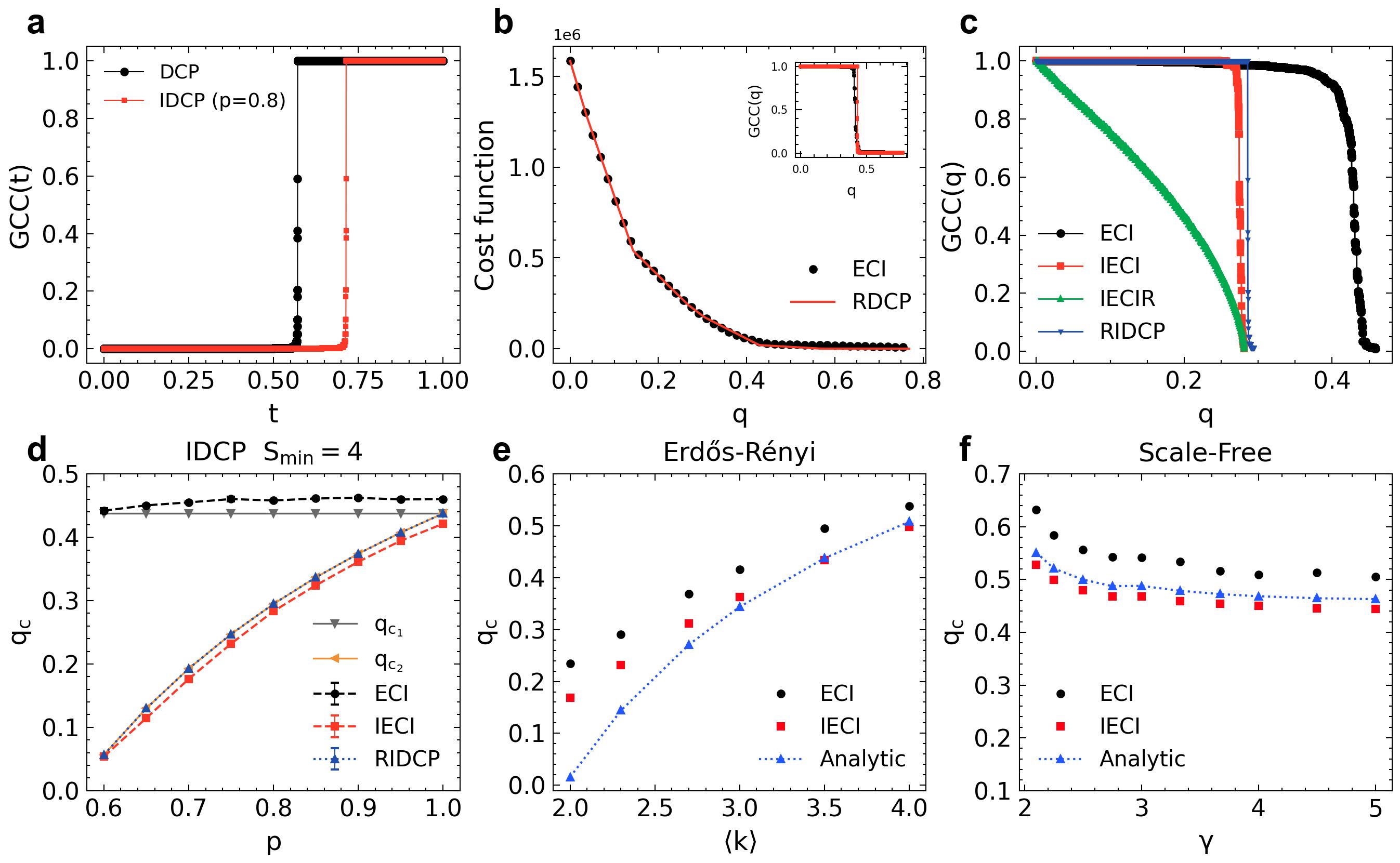}
    \caption{\textbf{Two explosive percolation models and their relation to the dismantling process.} \textbf{a,} $GCC(q)$ in the DCP process and the IDCP ($p=0.8$) process. 
    \textbf{b,} The descending process of the leading term $E_{\ell =1}$ of the cost energy function in the DCP network. ECI denotes the ECI dismantling process. RDCP denotes the reverse DCP process. Inset, $GCC(q)$ of each process in the DCP network.
    \textbf{c,} $GCC(q)$ for each process in the IDCP network with $S_{\min}=4$ and $p=0.8$. ECI, IECI and IECIR denote the dismantling processes of the corresponding algorithms. RIDCP denotes the reverse IDCP process.
    \textbf{d,} The corresponding dismantling thresholds of different processes in IDCP networks with different parameters $p$ ($N=10,000$, $M=17,500$, $S_{\min}=4$, degenerating to the ECI network if $p=1$, and error bars are s.e.m. over 20 realizations). RIDCP, ECI, and IECI represent the dismantling thresholds of the reverse IDCP, the ECI dismantling, and the IECI dismantling, respectively. $q_{c_1}$ and $q_{c_2}$ are the theoretical dismantling thresholds of the reverse DCP and the reverse IDCP, respectively. In all the cases here, the dismantling target is set to $C=0.01N$. \textbf{e, f,} The dismantling threshold $q_c$ of the ECI (black dot) and IECI (red square) algorithms and the analytic threshold ($q_{c_1}$, blue triangle) of the ECI algorithm versus average degree $\langle k \rangle$ (for Erdős–Rényi network) and degree exponent $\gamma$ (for scale-free network).}
    \label{fig:EP_total_new}
    \vspace*{-0.3cm}
\end{figure*}

To efficiently minimize $E_{\ell}(\mathbf{m})$, we propose an ECI algorithm: remove the most influential edge (with the highest ECI value) adaptively until the dismantling target ($ S_{GCC}\leq C$) is achieved. Specifically, the ECI algorithm dismantles the network step by step into chain-like structures, as shown in Fig.~\ref{fig:fig1}. Ideally, the network keeps connected until the critical point arrives, despite a continuous decrease in the density of network connectivity with the removal of edges. After reaching this critical point, the network undergoes a dramatic fragmentation, manifesting as systemic explosive vulnerability. The computational complexity of the ECI algorithm is $O(M\log M)$ by removing a finite proportion of edges at each step for $\ell = 1$ without a significant performance loss (see Supplementary Information Section~\ref{S1.4} for details). Hence, integrating the dismantling performance and time complexity, we adopt the ECI algorithm with $\ell =1$ unless otherwise specified. 

In practice, the network dismantling problem faces two distinct requirements under the same premise of dismantling target. One is the hidden dismantling, where $S_{GCC}$ scarcely decreases before the dismantling transition point, followed by a sharp drop. Here, to compensate for the effect of ignoring loop-like structures on ECI’s accuracy and enhance stealthiness at the transition point, we use the sum rule to improve the ECI algorithm, termed as improved ECI (IECI) algorithm (see Methods for details). With the sum rule, IECI involves greedily reinserting some edges into small clusters after ECI dismantling. Fig.~\ref{fig:model} showcases the best performance of ECI and IECI concerning hidden dismantling on synthetic networks, in particular their significant features of explosive dismantling and minimal removal cost, compared to all baseline strategies. IECI, through strategic reinsertion of secondary intra-cluster edges, identifies crucial inter-cluster edges more accurately, thereby further reducing the dismantling threshold and refining the stealthiness of ECI. 

The other is the fast dismantling, where $S_{GCC}$ diminishes swiftly throughout the whole dismantling process. For this purpose, we introduce an alternative approach, IECIR, based on IECI with a reordering strategy (see Methods for details). Specifically, IECIR reorders the edges identified by the IECI algorithm according to a particular rule, enabling a new removal sequence that facilitates a rapid decline of $S_{GCC}$. Notice that IECI and IECIR remove the same set of edges, just in a different order.

To validate the performance of IECI and IECIR on the END problem, we apply them to empirical networks and make comparison with classical dismantling algorithms and edge ranking algorithms. 
As shown in Fig.~\ref{fig:real}, given a dismantling target $C=0.01N$, IECI requires the removal of fewer edges in comparison to baseline methods and the reordering strategy additionally expedites the decline of the network's $S_{GCC}$. Furthermore, IECI yields a more homogeneous distribution of cluster sizes, leading to the emergence of numerous compact, isolated communities (Fig.~\ref{fig:real}b). This proves to be a valuable approach in the implementation of pandemic-related lockdowns, enabling the containment of global transmission while preserving the integrity of local communities and their dynamics. In addition, at an equivalent dismantling cost, IECIR demonstrates a superior ability to induce more effective fragmentation. (Fig.~\ref{fig:real}c ).

\section*{Relation between edge-based dismantling and reverse explosive percolation}

To unravel the underlying mechanism of ECI's explosive characteristic, we propose a dedicated explosive percolation model, the dual competitive percolation (DCP). Specifically, the DCP model starts with a null graph with $N$ nodes and then proceeds through two phases to add a total of $M (\geq N) $ edges. During Phase 1, at each step, the model selects the two smallest clusters in the current network and connects them by introducing an edge that links the nodes with the lowest degrees within each respective cluster. This iterative procedure persists until all nodes are connected, forming an unbranched chain composed of $N$ nodes and $N-1$ edges. During Phase 2, at each step, the model selects two nodes, both meeting the criterion that the sum of the generalized degrees of their neighboring edges is minimized, and subsequently introduces an edge connecting these selected nodes. This iterative procedure persists until the network comprises a total of $M$ edges. In light of the global competition in cluster size, the DCP process manifests a discontinuous transition at the percolation threshold~\cite{naglerImpactSingleLinks2011a}, see Fig.~\ref{fig:EP_total_new}a. Significantly, the DCP model introduces a competitive mechanism that extends beyond clusters to encompass individual nodes, setting it apart from the explosive percolation discussed in Ref.~\cite{naglerImpactSingleLinks2011a} \!\!. 

Specifically, let $t$ represent the fraction of added edges, and as the DCP process unfolds, $t$ increases from $0$ to $1$. Given the number of nodes $N=2^n (n\rightarrow \infty)$ and the average degree $\langle k\rangle$, the DCP process transitions to Phase 2 at $t=\frac{N-1}{M}$. Then the percolation threshold, at which a giant connected component emerges, is 
\begin{equation}
    t_{c_1}=\frac{N-1} {M} =\frac{N-1}{\frac{N\langle k\rangle}{2}}\approx \frac{2}{\langle k\rangle}, (N \rightarrow \infty).
\end{equation}

When reversing the DCP process, we can capture an explosive dismantling process, wherein a network with $N$ nodes and $M$ edges is systematically dismantled until reaching the predefined dismantling target $C$ (Fig.~\ref{fig:EP_total_new}b inset). At this juncture, the fraction of removed edges, referred to as the dismantling threshold, is approximately given by $q_{c_1}\approx TP_1+\frac{256}{N \langle k\rangle}$ (see Supplementary Information Section~\ref{SI:pt_dt} for details). Here, the dismantling transition point $TP_1$, marking the onset of significant network disintegration, is $TP_1=1-t_{c_1}\approx 1-\frac{2}{\langle k\rangle}$.

Similarly, we propose an improved dual competitive percolation (IDCP) model (see Methods for details), which increases the density of small communities within the network on the basis of the DCP model and is also an EP model with discontinuous transition~\cite{naglerImpactSingleLinks2011a} (Fig.~\ref{fig:EP_total_new}a). For the number of nodes $N=2^n (n\rightarrow \infty)$ and the average degree $\langle k\rangle$, according to the mean field theory~\cite{barabasiMeanfieldTheoryScalefree1999} \!\!, the percolation threshold $t_{c_2}$ of the IDCP satisfies $Mt_{c_2}=\frac{N-1}{p}$, so 
\begin{equation}
    t_{c_2}=\frac{N-1}{Mp}=\frac{N-1}{\frac{N\langle k\rangle}{2}p}\approx \frac{2}{\langle k \rangle p}, (N \rightarrow \infty).
\end{equation}
In addition, its dismantling transition point and dismantling threshold are $TP_2=1-t_{c_2} \approx 1-\frac{2}{\langle k \rangle p}$ and $q_{c_2}\approx TP_2+\frac{256}{N \langle k\rangle}$, respectively, for $C=0.01N$ (see Supplementary Information Section~\ref{SI:pt_dt} for details). 

During Phase 1, the DCP network consists of chain-like clusters due to its formation rule. Consequently, by combining the spectral properties of the non-backtracking matrix with the characteristics of the cost function in this chain-like network, we demonstrate that, during Phase 1, the DCP process ensures that the cost function $\left|\boldsymbol{w}_{\ell=1}(\mathbf{m})\right|^2$ of the network, i.e. $E_{\ell=1}(\mathbf{m})$ of the chain-like network, grows at the slowest speed at each step (see Methods for details). Conversely, the ECI algorithm endeavors to minimize $E_{\ell=1}(\mathbf{m})$ at the swiftest rate during each edge removal step. Consequently, at the dismantling threshold, the network attains a chain-like configuration, representing the network state characterized by the lowest cost function. Accordingly, the uniformly sized chain-like clusters produced by the ECI algorithm (Fig.~\ref{fig:real}b) align with the powder keg~\cite{friedmanConstructionAnalysisRandom2009} within the reverse DCP process. More precisely, their trajectories in reducing the cost function closely mirror each other, effectively illustrating that the reverse DCP process can faithfully replicate the ECI dismantling process (Fig.~\ref{fig:EP_total_new}b). Thus, the transparent edge addition rules employed by DCP offer insight into the intricate microscopic mechanisms that underlie the explosive characteristics of ECI.

The relation between the ECI dismantling process and the reverse DCP process reveals that the influential edges identified by ECI can cause a dramatic transition over the network's global functionality. Moreover, this relationship provides an approximate prediction of the ECI dismantling threshold $q_{c_{ECI}}$, further validating the robustness of ECI even in the presence of varying network models and topological parameters (Figs.~\ref{fig:EP_total_new}d, e and f). Notice that $q_{c_{ECI}}$ is slightly larger than $q_{c_1}$ due to the presence of loops within networks, as ECI may split some of them off from the network before reaching the dismantling transition point (see Supplementary Information Section~\ref{S5.} for in-depth discussion).

Benefiting from the sum rule, IECI presents two desirable outcomes. First, the dismantling process of IECI becomes sharper at the dismantling transition point, exhibiting enhanced stealthiness, as shown in Figs.~\ref{fig:model} and~\ref{fig:EP_total_new}c. Second, it effectively reduces the dismantling threshold, bringing it into closer alignment with analytical results, see Figs.~\ref{fig:EP_total_new}e and f. On the other hand, IDCP, by considering community structures that approximately correspond to clusters produced by the IECI algorithm (Fig.~\ref{fig:real}b), possesses the capability to delay the percolation threshold, as shown in Fig.~\ref{fig:EP_total_new}a. Its reverse process also closely replicates the IECI dismantling process, as depicted in Fig.~\ref{fig:EP_total_new}c. Thus, IDCP also effectively models IECI's explosive dismantling characteristic.

\section*{Discussion}

This study has adeptly tackled the identification of collective behaviors that manifest explosive vulnerability by introducing ECI and its refined version, IECI. These advancements have illuminated the intricate mechanisms underlying the eruption of large-scale disaster in networks, offering invaluable insights. Explosive vulnerability, due to its early concealment, erupts with exceptional intensity. ECI and IECI showcase the dual facets of optimal and explosive percolation. Subsequently, we propose the dual competitive percolation model, which generates explosive characteristic through explicit rules, revealing how collective behavior optimizes the network's cost energy function by inhibiting connectivity growth simultaneously at both cluster and node levels, highlighting the deeper intrinsic connection between optimal and explosive percolation.

This work signifies a substantial stride in the application of explosive percolation theory within the realm of network science, laying the foundation for comprehending and controlling explosive phenomena in complex networks. Specifically, differing from the assertion in Ref~\cite{clusellaImmunizationTargetedDestruction2016} \!\!, ECI demonstrates the feasibility of employing a universal optimal strategy based on a single score. In addition, both the ECI and IECI algorithms illustrate that systemic explosive vulnerability, potentially alongside other global features, arises from specific patterns of edge collective behavior. In contrast to individual edges, these distinct modes of collective behavior take on a more critical role. Once these modes are disrupted, systemic explosive vulnerability wanes, even if the same edges are targeted, as exemplified by IECIR. The efficacy demonstrated by IECIR in defusing explosive vulnerability outlines an effective route toward enhancing the security and robustness of such systems. In practical terms, our findings reverberate across diverse domains, including fortifying power grids against large-scale cascade failures, devising targeted vaccination strategies to curtail the spread of infectious diseases, and mitigating the pernicious propagation of misinformation within networked systems.

Looking ahead, comprehending and uncovering intricate dependencies among higher-order structures~\cite{battiston2021physics} will be pivotal in revealing explosive behavior within networks. Our extension of optimal percolation theory to edge structures using line graphs is expected to inspire its application to higher-order structures. Additionally, investigating whether different types of network topologies give rise to varying mechanisms of explosive vulnerability and how to tailor personalized dismantling strategies based on network topology features warrant further exploration. We anticipate that more research will build upon this foundation to delve deeper into this vital yet challenging aspect of network science.

% Main Text: 3220-301=2919 words

\FloatBarrier
\bibliographystyle{unsrt}
\bibliography{ECI_ref}

\renewcommand\thesection{M\arabic{section}}
\renewcommand\thesubsection{M\arabic{section}.\arabic{subsection}}
\setcounter{section}{0}
\setcounter{equation}{0}
\setcounter{figure}{0}
\renewcommand{\theequation}{M\arabic{equation}}
\renewcommand{\thefigure}{M\arabic{figure}}

{\footnotesize
\section*{Methods}

\emph{\underline{The sum rule for ECI}}. 
The sum rule is used to mitigate the detrimental impact of the loop-like structures on ECI algorithm. Specifically, given a set of edges $E_{r}$ removed by a dismantling algorithm and the dismantled network $G^*(V,E\backslash E_{r})$, let $SCC$ denote the set of connected components in network $G^*$, and set the sampling size $r$ as 100 based on the performance and complexity considerations of the sum rule (see Supplementary Information Section~\ref{S1.5}). The procedure of the sum rule is as follows:

\begin{enumerate}
\item Select $r$ pairs of adjacent connected components $r-CCs=\{(CC_{11},CC_{12}),\dots,(CC_{r1},CC_{r2})\}$ in $SCC$ that satisfy the merging condition (the sum of the number of nodes is not greater than $C$);

\item Compute the sum rule score $\boldsymbol{\alpha}=\{\alpha _1,\dots ,\alpha _r\}$ for each pair of connected components in $r-CCs$, which is defined as:
   \begin{equation}
       \alpha _i=\frac{nn_i}{adj(CC_{i1},CC_{i2})}, i\in \{1,2,\dots,r\}.
   \end{equation}
   $nn_i=nn(CC_{i1})+nn(CC_{i2})$, where $nn(CC_{i1})$ and $nn(CC_{i2})$ denote the number of nodes of the first and second connected component in the $i$th pair, respectively. $adj(CC_{i1},CC_{i2})$ denotes the number of edges between two connected components in the $i$th pair in the original network;

\item Select the pair $(CC_{j1},CC_{j2})$ with the smallest $\alpha _i$, merge these two connected components, and update $SCC$ and $E_{r}$;

\item Repeat steps 1-3 until no component satisfies the merging condition. The obtained $E_{r}$ is the final result.
\end{enumerate}

\emph{\underline{The reordering technique}}. 
In network science, we commonly use the robustness metric $R$ to qualify the performance of dismantling algorithms. For a given network $G(V,E)$, $R$ is defined as $R=\frac{1}{N}\sum _{m_r=1}^{M_{r}}S(m_r)$~\cite{schneider2011mitigation} \!\!, where $N=|V|$ is the total number of nodes in the network, $M_r$ is the total number of removed edges, and $S(m_r)$ is the value of $S_{GCC}$ after removing the first $m_r$ edges. A smaller $R$ indicates a more efficient algorithm and for a given $E_{r}$, the final $R$ will be different for different removal orders. Here, we propose a reordering technique for $E_{r}$ such that $S_{GCC}$ can drop as fast as possible, thus improving the dismantling efficiency. 

Suppose that $\Delta x$ edges are removed in certain steps of the dismantling process, then $S(m_r)$ becomes $S(m_r+\Delta x)$, so that the variation in $S_{GCC}$ is $\Delta y=S(m_r)-S(m_r+\Delta x)$. Then our goal is to remove the appropriate edges to maximize the score $\beta = \frac{\Delta y}{\Delta x}$, according to the definition of $R$. Specifically, given a network $G$ and the set of connected components $SCC$ in the remaining network $G^*(V,E\backslash E_{r})$ dismantled by one dismantling algorithm (e.g. IECI algorithm), we consider each connected component $CC_i$ in $SCC$ as a unit, and assume that to separate $CC_i$ with $\Delta y_i$ nodes from the original network $G$, we need to remove $\Delta x_i$ edges. Then the score $\beta _i$ of $CC_i$ is $\beta _i=\frac{\Delta y_i}{\Delta x_i}$. The reordering technique works as follows:

\begin{enumerate}
\item Compute the score $\beta _i$ of each $CC_i$ in $SCC$ in the current network $G$;
\item Separate the $CC_i$ with the largest $\beta _i$ from network $G$, and then update $SCC$ and $G$;
\item Repeat steps 1 and 2 until all connected components in $SCC$ are separated;
\item Reorder $E_r$ according to the separation order of connected components in $SCC$ to get the final removal order.
\end{enumerate} 

\emph{\underline{Improved dual competitive percolation}}.  
Given the number of nodes $N$ and the number of edges $M$, the improved dual competitive percolation (IDCP) starts with a network with $N$ nodes and 0 edges, and all nodes are initially divided into fixed communities of size $S_{\min}$. During Phase 1 of the IDCP process, at each step, either perform step 1 with probability $p$ or perform step 3 with probability $1-p$. Once all nodes are merged into one connected component, it turns to Phase 2. During Phase 2, at each step, either perform step 2 with probability $p$ or perform step 3 with probability $1-p$. The three steps are as follows:

\begin{enumerate}
\item Select the two smallest clusters in the current network and connect them by adding an edge between the two nodes with the smallest degrees in each cluster;
\item Add an edge randomly to the current network;
\item Randomly select one of the communities and randomly add an edge to it.
\end{enumerate}

Note that the division of nodes into fixed communities actually occurs during step 1: when the size of each connected component in the network simultaneously reaches or exceeds $S_{\min}$ for the first time, we regard these components as fixed communities in the network. Here, to ensure that a certain number of nodes are added in step 3, the IDCP model assumes that this division is pre-determined. Finally, there are $N-1$ edges added in step 1, $Mp-N+1$ edges added in step 2, and $M(1-p)$ edges added in step 3. The size of fixed communities $S_{\min}$ is generally set as 4, and $p$ determines the average number of edges within these communities.

\emph{\underline{Lemma and theorem}}. 

Below are two related lemmas, followed immediately by a proof of a theorem on the variation of the cost function of the DCP network during Phase 1. 
It is important to note that our proof utilizes the topological correspondence between network $G$ and its line graph $L(G)$~\cite{hararySome1960} \!\! (see Supplementary Information Section~\ref{S1.3} for details). Each node of $L(G)$ represents the corresponding edge of $G$, and two nodes of $L(G)$ are adjacent if and only if their corresponding edges share a common endpoint in $G$.
Lemma 1 shows the topological characteristic of the DCP network in Phase 1. Lemma 2 shows that the cost function of a chain-like network is always smaller than that of a network whose line graph is loop-containing, which provides a theoretical basis for the finitization of $\ell$. With the help of these two lemmas, we can easily prove Theorem 1, demonstrating that the cost function $\left|\boldsymbol{w}_{\ell=1}(\mathbf{m})\right|$ of the DCP network grows the slowest during Phase 1.

\begin{lemma}
    Prior to the completion of Phase 1, the network constructed by the DCP model exclusively consists of isolated nodes and unbranched chains.
\end{lemma}
\begin{proof}
    At the beginning of the DCP process, there are only isolated nodes in the network, and connecting isolated nodes will generate the simplest chain-like components. After this, during Phase 1, the DCP process connects the nodes with the smallest degree (i.e., leaf nodes) of different clusters each time, so the newly connected clusters are still chain-like clusters. Therefore, this lemma holds.
\end{proof}

\begin{lemma}\label{lem2}
    The cost function $\left|\boldsymbol{w}_{\ell}(\mathbf{1})\right|$ of a network whose line graph is loop-containing is larger than that of a chain-like network.
\end{lemma}
\begin{proof}
    The cost function of the network characterizes the magnitude of the maximum eigenvalue of the non-backtracking matrix of the corresponding line graph (see Supplementary Information Section~\ref{S1.3}). The eigenvalue of the non-backtracking matrix of the loop-containing network is greater than or equal to 1, while the eigenvalue of the non-backtracking matrix of the tree network is 0~\cite{gloverSpectralPropertiesNonbacktracking2021} \!\!. Therefore, the cost function of a chain-like network, whose line graph is a tree network, is smaller than that of a network whose line graph is loop-containing. 
\end{proof}

Remarkably, the line graphs of tree networks with multiple branches and loop-containing networks are loop-containing (see Supplementary Information Fig.~\ref{fig:S0}). For chain-like networks, since the eigenvalue of the edge-based non-backtracking matrix is 0, i.e., $\left|\boldsymbol{w}_{\ell}(\mathbf{m})\right|^2$ is 0 as $\ell \rightarrow \infty$, we consider the cost energy function $\left|\boldsymbol{w}_{\ell}(\mathbf{m})\right|^2$ of the finite case of $\ell$ ($\ell =1$).

\begin{theorem}\label{thm1}
    For $\ell=1$, the cost function $\left|\boldsymbol{w}_{\ell}(\mathbf{m})\right|^2$ of the DCP network grows the slowest during Phase 1.
\end{theorem}

\begin{proof}
    According to Lemma 1, the DCP network is a chain-like network during Phase 1, and the cost function $\left|\boldsymbol{w}_{\ell=1}(\mathbf{m})\right|^2$ is simplified to its leading term $E_{\ell=1}$. According to Lemma 2, the cost function of the DCP network during Phase 1 is smaller than that of other networks whose line graph is loop-containing. (1) For the DCP process, each added edge during Phase 1 connects the leaf nodes from the two smallest clusters. (2) For the leaf node $i$, its $\partial \operatorname{Ball}(i,\ell=1)$ contains the smallest number of nodes and its degree is the smallest. These two points indicate that the increment of the cost function $E_{\ell=1}$ resulting from connecting leaf nodes in the chain structure is minimal. Thus, the DCP process makes the growth of $E_{\ell=1}$, also $\left|\boldsymbol{w}_{\ell=1}(\mathbf{m})\right|^2$ for chain-like networks, the slowest during Phase 1.
\end{proof}

%TC:ignore
% \newpage
\textbf{\small Code availability}. All codes supporting our findings are available from the GitHub repository: \url{https://github.com/PPNew1/Edge_Collective_Influence}.

\textbf{\small Data availability}.
All data supporting our findings are available from the corresponding author upon reasonable request.

\textbf{\small Acknowledgments}. 
The authors are grateful for the support from the STI 2030--Major Projects (2022ZD0211400), the National Natural Science Foundation of China (Grant No. T2293771), the China Postdoctoral Science Foundation (2022M710620), Sichuan Science and Technology Program (2023NSFSC1919, 2023NSFSC1353), the Project of Huzhou Science and Technology Bureau (2021YZ12), the UESTCYDRI research start-up (U032200117), Young Leading Talents of Nantaihu Talent Program in Huzhou (2023) and the New Cornerstone Science Foundation through the XPLORER PRIZE.

\textbf{\small Author contributions}. 
Peng Peng, Tianlong Fan, Xiao-Long Ren and Linyuan Lü conceived the research and designed the experiments. Peng Peng developed the mathematical framework, carried out the numerical experiments, worked out the theory and wrote the draft with the help of Tianlong Fan and Xiao-Long Ren. Tianlong Fan, Xiao-Long Ren and Linyuan Lü edited the paper. All authors analyzed the data and discussed the results. The authors would like to thank Dr. Shuqi Xu for insightful discussions and suggestions.}

%Methods: 4417-3220=1197 words

% --------------------------------------------------------------------------------------------------
% ************* SUPPLEMENTARY MATERIAL STARTS HERE ***************
% --------------------------------------------------------------------------------------------------

\newpage
\onecolumngrid

\titleformat{\section}{\filcenter\normalfont\bfseries}{\Roman{section}.}{1em}{}
\titleformat{\subsection}{\filcenter\normalfont\bfseries}{\Alph{subsection}.}{1em}{}

\renewcommand\thesection{\Roman{section}}
\renewcommand\thesubsection{\thesection\Alph{subsection}}

\setcounter{section}{0}
\setcounter{equation}{0}
\setcounter{figure}{0}
\setcounter{page}{1}
\renewcommand{\theequation}{S\arabic{equation}}
\renewcommand{\figurename}{\textsc{Extended Data Figure}}
\renewcommand{\thefigure}{ED\arabic{figure}}
\renewcommand{\thepage}{S\arabic{page}}

\newpage

\begin{quote}
\centering 
{\large \bf Unveiling Explosive Vulnerability of Networks through Edge Collective Behavior}\vspace*{+0.1cm}\\
{\normalsize (\underline{\textsc{Supplementary Information/Extended Data}})}\vspace*{+0.25cm}\\
{\textsc{Peng Peng$^{1}$, Tianlong Fan$^{2}$, Xiao-Long Ren$^{3}$, Linyuan L\"{u}$^{2,1}$}}\\
{\small \em $^{1}$Institute of Fundamental and Frontier Studies, University of Electronic Science and Technology of China, 611731 Chengdu, P.R. China}\\
{\small \em $^{2}$School of Cyber Science and Technology, University of Science and Technology of China, 230026 Hefei, P.R. China}\\
{\small \em $^{3}$Yangtze Delta Region Institute (Huzhou), University of Electronic Science and Technology of China, 313001 Huzhou, P.R. China}\\
{\small (Dated: \today)}
\end{quote}

\section{Details of edge collective influence and its expansion}\label{S1.}

Most past dismantling algorithms~\citesec{moroneInfluenceMaximizationComplex2015,braunsteinNetworkDismantling2016,clusellaImmunizationTargetedDestruction2016,mugishaIdentifyingOptimalTargets2016,zdeborovaFastSimpleDecycling2016,tianArticulationPointsComplex2017,liuOptimizationTargetedNode2018} focused on the node-based network dismantling problem, aiming to reduce the size of the giant connected component (GCC) in the network by removing the least number of nodes. However, Ren et al.~\citesec{renGeneralizedNetworkDismantling2019} recognized that the cost of removing nodes varies, with nodes adjacent to more edges costing more, and proposed a generalized network dismantling (GND) algorithm to fragment the network with minimal removal cost. This algorithm minimizes the number of edges removed when removing nodes, thereby reducing the overall cost of the dismantling process. Here, we focus directly on the edge-based network dismantling problem (END) to develop a cost-efficient dismantling algorithm. Although END, like GND, also takes the number of removed edges as the cost of network dismantling, it differentiates from GND in its approach: while GND removes all the edges between the attacked node and its neighboring nodes at each time step, END removes only one edge at each time step, and the edges removed at adjacent time steps are not necessarily related. Specifically, given a network $G(V,E)$, where $V$ and $E$ denote the set of nodes and the set of edges in the network respectively, and $N=|V|$ and $M=|E|$ are the number of nodes and edges of the network respectively, then END can be formulated as follows: given a dismantling target $C$, find a minimum set of edges $E_{r}$ such that the size $S_{GCC}$ of the GCC in the network $G^*(V,E\backslash E_{r})$ satisfies $S_{GCC} \leq C$.

\subsection{Optimal percolation theory}\label{S1.1}

Let the vector $\mathbf{m}=\{m_1,\dots,m_M\}$ denote whether the edges in the network are removed or not, where $m_i=0$ means that edge $i$ is removed, otherwise $m_i=1$. The fraction of removed edges is $q=1-\frac{\sum_{i}m_i}{M}$. As shown in Fig.~\ref{fig:S0}, for a given network $G(V,E)$, we can transform it into a line graph $L(G)$. Each node of $L(G)$ represents the corresponding edge of $G$, and two nodes of $L(G)$ are adjacent if and only if their corresponding edges share a common endpoint in $G$~\citesec{hararySome1960} \!\!. Here,  we map edges to nodes using the mapping of the network $G$ to the line graph $L(G)$. Let $A_m=\{a_{ij}\}$ denote the edge adjacency matrix of network $G$, where $a_{ij}=1$ if edge $i$ is adjacent to edge $j$, otherwise $a_{ij}=0$. We define the generalized degree of edge $i$ as $k_i=\sum_j a_{ij}$, which denotes the number of edges adjacent to $i$. 

\begin{figure}[H]
  \centering
  \includegraphics[width=0.8\textwidth]{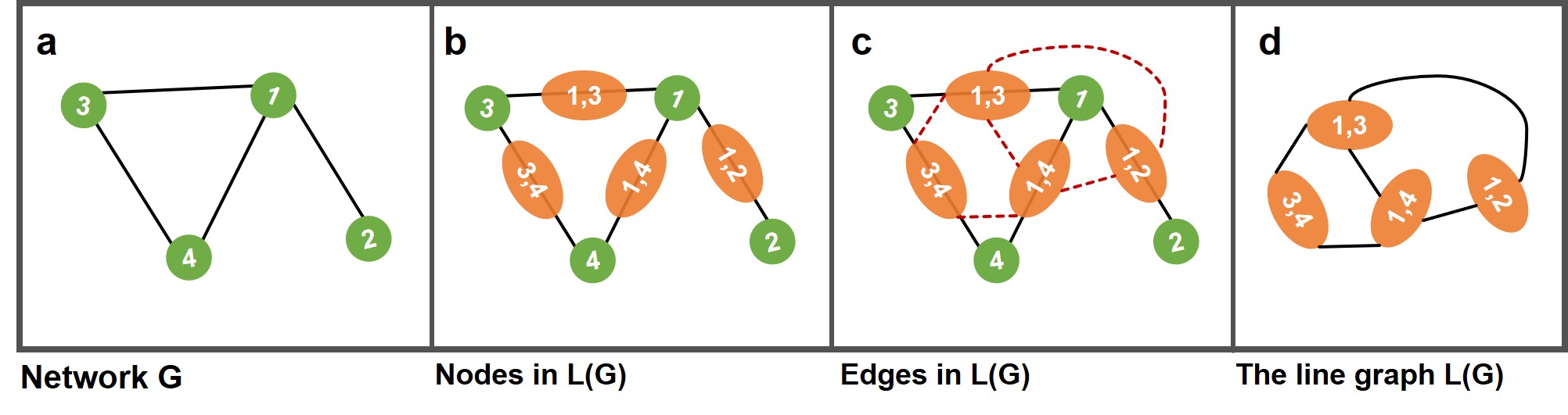}
  \caption{\centering {The transformation process from network $G$ to line graph $L(G)$.}}
  \label{fig:S0}
\end{figure}

For END, the order parameter of an edge's influence is the probability $e_{i \rightarrow j}$ that edge $i$ still belongs to the GCC after removing edge $j$. When the dismantling target $C\rightarrow 0$, it is required that $\sum_{i,j} e_{i \rightarrow j}$ becomes sufficiently small ~\citesec{moroneInfluenceMaximizationComplex2015} \!\!. If edge $j$ is temporarily removed from the network, the probability that edge $i$ belongs to the GCC depends on whether any of its neighboring edges (except for $j$) still belong to the GCC after removing $i$. For a network whose corresponding line graph is locally tree-like, this can be expressed as the following message passing formula~\citesec{newmanMessagePassingMethods2023} \!\!:
\begin{equation}
    e_{i \rightarrow j}=m_i\left[1-\prod_{k \in \partial i \backslash j}\left(1-e_{k \rightarrow i}\right)\right],
    \label{eq:SI 1}
\end{equation}
where $\partial i \backslash j$ is the set of edges adjacent to edge $i$ except for $j$. 

When the dismantling target $C\rightarrow 0$, it is required that the probability that a randomly selected node belongs to the GCC is small enough after removing some of the edges, which is consistent with $\sum_{i,j} e_{i \rightarrow j}\rightarrow0$. Thus, we need a proper $\mathbf{m}$ to satisfy $\sum_{i,j} e_{i \rightarrow j}\rightarrow0$. Consider the solution for all edges $i,j$: $\{e_{i \rightarrow j}=0\}$, whose stability is determined by the maximum eigenvalue $\lambda(\mathbf{m};q)$ of the linear operator $\hat{\mathcal{M}}$, where $\hat{\mathcal{M}}=\{\mathcal{M}_{k\rightarrow l,i\rightarrow j}\}$, $\mathcal{M}_{k\rightarrow l, i\rightarrow j}\equiv \frac{\partial e_{i\rightarrow j}}{\partial e_{k\rightarrow l}}|_{\{e_{i\rightarrow j}=0\}}$ ($e_{i\rightarrow j}=0$ when edge $i$ does not belong to the GCC after edge $j$ is removed). For the network whose corresponding line graph is locally tree-like, $\hat{\mathcal{M}}$ can be represented by the modified form of the edge-based non-backtracking operator $\hat{\mathcal{B}}$~\citesec{gloverSpectralPropertiesNonbacktracking2021} \!\!:
\begin{equation}
   \mathcal{M}_{k\rightarrow \ell,i\rightarrow j}=m_i\mathcal{B}_{k\rightarrow \ell,i\rightarrow j},
   \label{eq:SI 2}
\end{equation}
where
\begin{equation}
    \mathcal{B}_{k \rightarrow \ell, i \rightarrow j}=\left\{\begin{array}{l}
    1 \text { if } \ell=i \text { and } j \neq k, \\
    0 \text { otherwise. }
    \end{array}\right.
    \label{eq:SI 3}
\end{equation}

Let the first-order moment of the nodal degree distribution of the network $G(V,E)$ be $\langle k\rangle$, the second-order moment be $\langle k^2\rangle$, and we find that $\hat{\mathcal{M}}$ is an $M_l\times M_l$ matrix, where $M_l=(\langle k^2\rangle-\langle k\rangle)N$.

The solution $\{e_{i \rightarrow j}=0\}$ is stable when $\lambda(\mathbf{m};q)<1$. Thus, for a given $q$, to find the optimal solution $\mathbf{m}^\ast$ of END, we need to find the optimal edge removal configuration $\mathbf{m}^\ast$ such that $\lambda(\mathbf{m};q)$ is the smallest among all configurations that satisfy $\langle m\rangle = 1-q$. That is
\begin{equation}
    \lambda\left(\mathbf{m}^{*} ; q\right) \equiv \min _{\mathbf{m}:\langle m\rangle=1-q} \lambda(\mathbf{m} ; q).
    \label{eq:SI 4}
\end{equation}

The optimal threshold $q_c$ is the solution of the equation:
\begin{equation}
    \lambda\left(\mathbf{m}^{*} ; q_{c}\right)=1.
    \label{eq:SI 5}
\end{equation}

\subsection{Cost energy function of edge influence}\label{S1.2}

In the following, we fix $q$ and simplify $\lambda(\mathbf{m};q)$ to $\lambda(\mathbf{m})$. Based on the power method, $\lambda(\mathbf{m})$ can be represented as follows:
\begin{equation}
    \lambda(\mathbf{m})=\lim _{\ell \rightarrow \infty}\left[\frac{\left|\boldsymbol{w}_{\ell}(\mathbf{m})\right|}{\left|\boldsymbol{w}_0\right|}\right]^{1 / \ell},
    \label{eq:SI 6}
\end{equation}
where $\boldsymbol{w}_0$ is an arbitrary fixed vector and $\left|\boldsymbol{w}_{\ell}(\mathbf{m})\right|$ is given by the following equation:
\begin{equation}
    \left|\boldsymbol{w}_{\ell}(\mathbf{m})\right|=\left\langle\boldsymbol{w}_{\ell} \mid \boldsymbol{w}_{\ell}\right\rangle^{\frac{1}{2}}=\left|\hat{\mathcal{M}}^{\ell} \boldsymbol{w}_0\right|=\left\langle\boldsymbol{w}_0\left|\left(\hat{\mathcal{M}}^{\ell}\right)^{\dagger} \hat{\mathcal{M}}^{\ell}\right| \boldsymbol{w}_0\right\rangle^{\frac{1}{2}}\sim \mathrm{e}^{\ell \log \lambda(\mathbf{m})}.
    \label{eq:SI 7}
\end{equation}

To minimize $\lambda(\mathbf{m})$, we need to find the optimal configuration $\mathbf{m}$ minimizing $\left|\boldsymbol{w}_{\ell}(\mathbf{m})\right| (\ell\rightarrow\infty)$. In practice, we can only minimize the edge influence cost energy function $\left|\boldsymbol{w}_{\ell}(\mathbf{m})\right|^2$ for a finite $\ell$~\citesec{moroneInfluenceMaximizationComplex2015} \!\!. To satisfy the locally tree-like assumption of our entire approach, we can ignore the influence of loops in the corresponding line graph of the original network and approximate the local environment around any mapping node by a tree. Hence, the leading term of the cost function is approximated as~\citesec{moroneInfluenceMaximizationComplex2015} \!\!:
\begin{equation}
    E_{\ell}(\mathbf{m})=\sum_{i=1}^{M} z_{i} \sum_{j \in \partial \operatorname{Ball}(i, \ell)}\left(\prod_{k \in \mathcal{P}_{\ell}(i, j)} m_{k}\right) z_{j},
    \label{eq:SI 8}
\end{equation}
where $z_i=k_i-1$, $\ell$ is the length of the shortest path from edge $i$ to $j$, $\operatorname{Ball}(i, \ell)$ is the set of edges $j$ inside the ball with radius $\ell$ at the center of edge $i$, $\partial \operatorname{ Ball}(i, \ell)$ is the boundary of the ball, and $\mathcal{P}_{\ell}(i, j)$ is the shortest path of length $\ell$ between $i$ and $j$.

\subsection{Correspondence between network $G$ and its line graph $L(G)$}\label{S1.3}

Based on the definition of the line graph, we find the following correspondence between network $G$ and its line graph $L(G)$. The node with degree $k+1$ in $G$ corresponds to the $k$-simplex ( a complete subgraph with $k$ nodes) in $L(G)$ , and the edge adjacency matrix $A_m$ of $G$ corresponds to the adjacency matrix $A_{L(G)}$ of $L(G)$, so that the modified edge-based non-backtracking operator $\hat{\mathcal{M}}$ of $G$ corresponds to the modified non-backtracking operator $\hat{\mathcal{M}}_{L(G)}$ of $L(G)$. Thus the edge influence cost energy function $\left|\boldsymbol{w}_{\ell}(\mathbf{m})\right|^2$ of $G$ corresponds to the node influence cost energy function $\left|\boldsymbol{w}_{\ell}(\mathbf{n_{L(G)}})\right|^2$ of $L(G)$, and their approximations $E_{\ell}(\mathbf{m})$ and $E_{\ell}(\mathbf{n_{L(G)}})$ of the leading terms also correspond to each other. Table~\ref{tab:tab1} shows the corresponding elements in network $G$ and its line graph $L(G)$. Accordingly, the most influential node identified by CI in L(G) obviously corresponds to the most influential edge identified in $G$, as illustrated in Fig.~\ref{fig:ECI_SI}.

\begin{table}[H]
\centering
\caption{\centering {Correspondence between $G$ and $L(G)$}}
\label{tab:tab1}
\begin{tabular}{l|l}
\hline
$G$                                              & $L(G)$                                                             \\ \hline
Edge                                             & Node                                                             \\
Node with degree $k+1$                                 & $k$-simplex                                                       \\
$A_m$                                            & $A_{L(G)}$                                                       \\
$\hat{\mathcal{M}}$                              & $\hat{\mathcal{M}}_{L(G)}$                                       \\
$\left|\boldsymbol{w}_{\ell}(\mathbf{m})\right|^2$ & $\left|\boldsymbol{w}_{\ell}(\mathbf{\mathbf{n_{L(G)}}})\right|^2$ \\
$E_{\ell}(\mathbf{m})$                           & $E_{\ell}(\mathbf{n_{L(G)}})$                                    \\ \hline
\end{tabular}
\end{table}

\begin{figure}[H]
  \centering
  \includegraphics[width=0.5\textwidth]{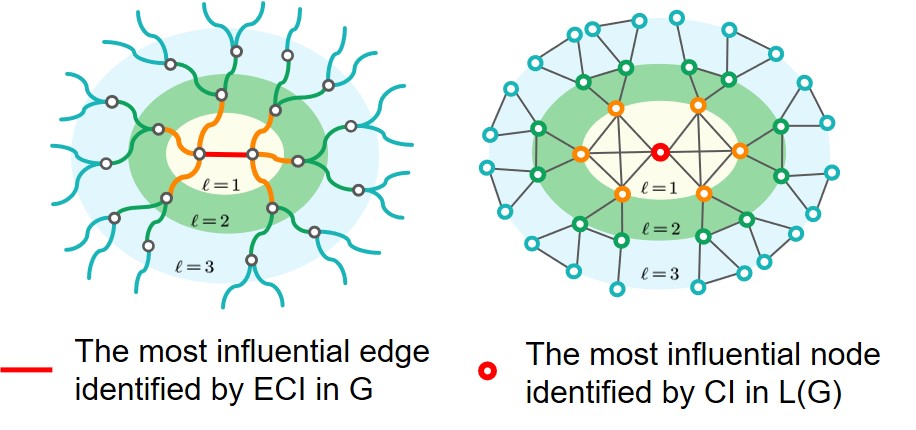}
  \caption{\centering {Correspondence between ECI in $G$ and CI in $L(G)$.}}
  \label{fig:ECI_SI}
\end{figure}

\subsection{Scalability of the ECI algorithm}\label{S1.4}

Similar to the CI algorithm, the complexity of the ECI algorithm is $O(M^2\log M)$ for finitely small $\ell$, and can be reduced to $O(M\log M)$ without loss of performance by removing a certain percentage of edges at each step~\citesec{moroneInfluenceMaximizationComplex2015} (Fig.~\ref{fig:S1}). The network density $d(L(G))$ of the line graph $L(G)$ corresponding to network $G$ is:
\begin{equation}
    d(L(G))=\frac{\sum_k \frac{k(k-1)}{2}p_k N}{M}=\frac{N}{M} \sum_k \frac{k^2-k}{2} p_k =\frac{\langle k^2 \rangle -\langle k \rangle}{\langle k \rangle}.
\end{equation}
Thus, when network $G$ exhibits the power-law degree distribution found in real networks, the density of $L(G)$ will be extremely large, which will lead to a dramatic increase in the time complexity of the ECI algorithm as $\ell$ increases. Therefore, we finally adopt $\ell=1$, after trading off the dismantling results of the ECI algorithm (Fig.~\ref{fig:S2}) with the complexity of the algorithm under different $\ell$.

\subsection{Effect of different sampling numbers $r$ on the sum rule}\label{S1.5}

Different $r$ implies different sampling numbers, which to some extent determines the improvement effect of the sum rule. In Fig.~\ref{fig:S3} we compare the dismantling results of the IECIR algorithm with different $r$. The dismantling performance at $r=100$ is comparable to that at $r=200$, even though the number of communities being compared at each step is only half. Therefore, reinsertion with $r=100$ is used throughout the main text.

\section{Network data}\label{S2.}

The real-world networks used in the main text consist of a power-grid network~\citesec{rossiNetworkDataRepository2015} (Power), a scientific collaboration network~\citesec{rossiNetworkDataRepository2015} (Collaboration), an Email network~\citesec{rossiNetworkDataRepository2015} (Email), and an online social network of Petster–Hamster~\citesec{renGeneralizedNetworkDismantling2019} (Social). The applications in these four networks can basically capture the real-world potential of our algorithms across various networks, encompassing tasks such as the maintenance of power grid systems, communication enhancement between scientific communities, curbing the malicious spread of disinformation in Email networks, and blocking epidemic outbreaks within society. The following table shows the topology information of these four networks.

\begin{table}[H]
\centering
\caption{\centering{Network Parameters}}
\begin{tabular}{l|llll}
\hline
          & $N$    & $M$     & $\langle k\rangle$ & $\langle k^2\rangle$ \\ \hline
Collaboration & 5094 & 7515  & 2.95               & 48.17                \\
Power         & 4941 & 6594  & 2.67               & 10.33                \\
Social        & 2000 & 16098 & 16.10              & 704.71               \\
Email         & 1134 & 5451  & 9.61               & 179.82              \\
\hline
\end{tabular}
\end{table}

\section{Baseline methods}\label{S3.}

\noindent \textbf{Bridgeness(BG)}~\citesec{chengBridgenessLocalIndex2010} 

Bridgeness is a local index used to measure the significance of an edge in maintaining the global connectivity of a network. For edge $e(u,v)$ , its ability to connect the large clusters in the network is defined as:

\begin{equation}
    B_{e(u, v)}=\frac{\sqrt{S_u S_v}}{S_{e(u, v)}},
    \label{eq:SI 15}
\end{equation}

\noindent where $S_u$, $S_v$ and $S_{e(u, v)}$ are the sizes of the largest clusters containing node $u$, node $v$ and edge $e(u, v)$, respectively.

\noindent \textbf{Edge Betweenness(EB)}~\citesec{girvanCommunityStructureSocial2002} 

Edge Betweenness is a centrality measure used to identify the edges that are most important for maintaining the global connectivity of a network. According to the idea that the more shortest paths between pairs of nodes pass through an edge $e(u, v)$, the more important the edge $e(u, v)$ is, the centrality of edge $e(u, v)$ is defined as

\begin{equation}
    EB(u, v)=\sum_{s \neq t \in V} \frac{\delta_{s t}(u, v)}{\delta_{s t}},
    \label{eq:SI 16}
\end{equation}

\noindent where $\delta_{st}$ is the number of all shortest paths between node $s$ and node $t$, and $\delta_{st}(u, v)$ is the number of all shortest paths between node $s$ and node $t$ that pass through edge $e(u, v)$. A larger EB score means greater importance of the edge.

\noindent \textbf{Collective Influence(CI)}~\citesec{moroneInfluenceMaximizationComplex2015} 

The CI algorithm identifies the smallest set of influencers by solving the node-based optimal percolation problem. 
Given the parameter $\ell $, the value of $CI_{\ell} (u)$ of node $u$ in the network is defined as follows:

\begin{equation}
    \mathrm{CI}_{\ell}(u)=\left(k_{u}-1\right) \sum_{v \in \partial \mathrm{Ball}(u, \ell)}\left(k_{v}-1\right),
    \label{eq:SI 9}
\end{equation}

\noindent where $k_u$ is the degree of node $u$ and $\partial \mathrm{Ball} (u,\ell)$ denotes the set of nodes in the network whose shortest path length to $u$ is $\ell$.

\noindent \textbf{Explosive Immunization(EI)}~\citesec{clusellaImmunizationTargetedDestruction2016} 

The EI algorithm is a method used to interrupt the spread of infection in a network. It combines the explosive percolation (EP) paradigm with the idea of maintaining a fragmented distribution of clusters. This algorithm heuristically utilizes two node scores, $\sigma_u^{(1)}$ and $\sigma_u^{(2)}$, to estimate a node's ability to interrupt infection propagation at two phases, which are defined as follows:

\begin{equation}
    \sigma_u^{(1)}=k_u^{(\mathrm{eff})}+\sum_{\mathcal{C}\subset \mathcal{N}_u}(\sqrt{|\mathcal{C}|}-1).
    \label{eq:SI 10}
\end{equation}

\noindent The first term $k_u^{\text {(eff) }}$ is the effective degree of node $u$, which is determined self-consistently from the original degree $k_u$:
\begin{equation}
    k_u^{(\mathrm{eff})}=k_u-L_u-M_u\left(\left\{k_v^{(\mathrm{eff})}\right\}\right),
    \label{eq:SI 11}
\end{equation}

\noindent where $L_u$ and $M_u$ are the number of leaf and hub nodes in the vicinity of $u$ respectively. During the iteration, nodes with effective degree $k_v^{(\mathrm{eff})} \geq K$ are regarded as hub nodes for a suitably chosen constant $K$. The second term is determined by the size $|\mathcal{C}|$ of cluster $\mathcal{C}$ in the set $\mathcal{N}_u$ of all clusters linked to $u$.

As the percolation process proceeds, some harmful nodes identified by $\sigma_u^{(1)}$ become harmless. To distinguish the influence of nodes more accurately, the EI algorithm uses $\sigma_u^{(2)}$ to evaluate the influence of nodes in Phase 2.

\begin{equation}
    \sigma_u^{(2)}= \begin{cases}\infty & \text { if } \mathcal{G}(q) \not \subset \mathcal{N}_u, \\ \left|\mathcal{N}_u\right| & \text { else, if } \arg \min _u\left|\mathcal{N}_u\right| \text { is unique, } \\ \left|\mathcal{N}_u\right|+\epsilon\left|\mathcal{C}_2\right| & \text { else. }\end{cases}
    \label{eq:SI 12}
\end{equation}

\noindent Here $\mathcal{G}(q)$ is the largest cluster as $qN$ nodes are removed, $|\mathcal{N}_u|$ is the number of clusters in the neighborhood of $u$, $\mathcal{C}_2$ is the second-largest cluster in $\mathcal{N}_u$, and $\epsilon$ is a small positive number (its value is not important provided $\epsilon \ll \frac{1}{N}$).
See the original article~\citesec{clusellaImmunizationTargetedDestruction2016} for more details about the EI algorithm.

\noindent \textbf{Generalized Network Dismantling(GND)}~\citesec{renGeneralizedNetworkDismantling2019} 

GND is a method designed to fragment the network into subcritical network components with minimal removal cost. It is based on the spectral properties of the node-weighted Laplacian operator $L_w$ , and thereby transforms the generalized network dismantling problem into an integer programming problem as shown below. It achieves high performance by combining the approximation spectrum of the Laplace operator with a fine-tuning mechanism associated with the weighted vertex cover problem.

\begin{equation}
    \min_{\boldsymbol{x}=\{x_1,x_2,\dots,x_n\}} \frac{1}{4} \boldsymbol{x}^\top L_w \boldsymbol{x}
    \label{eq:SI 13}
\end{equation}

\noindent subject to

\begin{equation}
    \begin{gathered}
    \boldsymbol{1}^\top \boldsymbol{x}=0, \\
    x_i \in\{+1,-1\}, i \in\{1,2, \ldots, n\} .
    \end{gathered}
    \label{eq:SI 14}
\end{equation}

\noindent To further improve the performance of GND, Ren et al. applied the reinsertion method and proposed GND with reinsertion (GNDR) ~\citesec{renGeneralizedNetworkDismantling2019} \!\!.
See the original article~\citesec{renGeneralizedNetworkDismantling2019} for more details about GND and its improved version GNDR.

\section{Dismantling transition point and dismantling threshold}\label{SI:pt_dt}

For reverse processes of the DCP model and the IDCP model, the dismantling transition point ($TP$) refers to the point where $S_{GCC}$ starts to drop rapidly, and the dismantling threshold ($q_c$) refers to the fraction of edges removed when reaching the dismantling target $S_{GCC} \leq C (C=0.01N)$. Therefore, $q_c > TP$. We denote the difference between them as $\Delta q = q_c - TP$. For both the DCP and IDCP networks with $N=2^n$ and $M=\frac{N\langle k\rangle}{2}$, suppose we need to further remove $\Delta m$ edges to reach the dismantling target after $TP$. Since each cluster in the network satisfies the dismantling target when its size is $2^{n-7}$, $\Delta m$ can be approximated as Eq.~\eqref{eq:delta m}, and $\Delta q = \frac{\Delta m}{M} \approx \frac{1}{2^{n-8}\langle k \rangle}$. Therefore, $q_{c}\approx TP+\frac{1}{2^{n-8} \langle k\rangle}=TP+\frac{256}{N\langle k\rangle}$.

\begin{equation}
    \begin{aligned}
        \Delta m &= (N-1)-(\frac{1}{2}N+\frac{1}{4}N+\dots+\frac{1}{2^{n-7}}N)\\
        &=\frac{1}{2^{{n-7}}}N-1
        \label{eq:delta m}
    \end{aligned}
\end{equation}

For the IDCP network, the dismantling threshold of the ECI algorithm, $q_{c_{ECI}}$, hardly varies with parameter $p$ and parameter $S_{\min}$, while the dismantling threshold of the IECI algorithm, $q_{c_{IECI}}$, varies significantly across different parameter cases (Fig.~\ref{fig:S4}).
Visibly, as the community of the network becomes sparser, i.e., $p$ or $S_{\min}$ increases, $q_{c_{IECI}}$ will increase because the improvement of the sum rule on the ECI algorithm will be diminished.

\section{How do the ECI and IECI algorithms work?}\label{S5.}

As shown in the main text, the ECI dismantling process, in reality, is not exactly identical to the reverse DCP process: $S_{GCC}$ will decrease earlier and achieve the dismantling target later. Figs.~\ref{fig:S5} a and b show that when confronted with a particular network structure, ECI will not prioritize dismantling loops or higher-order structures into chains, but will first separate them from the network, which will lead to an early drop of $S_{GCC}$. This is because the bridging edges, identified as the most crucial nodes in their line graph by ECI, are prioritized for removal, leading to disconnections. Such edges are referred to as ``weak edges". In addition, when the network is dismantled into chain-only clusters, there are many edges with the same maximum collective influence identified by ECI, and different removal orders of them will yield various dismantling results (Fig.~\ref{fig:S5}c). 

The IECI algorithm basically compensates for the shortcomings of ECI by the sum rule. The more obvious the small communities, the better the ECIR algorithm performs (Fig.~\ref{fig:S4}). This is primarily because smaller community sizes tend to match the genuine local clustering structures in the network, especially for relatively sparse networks.

%Bibliography
\bibliographystylesec{unsrt}
\bibliographysec{ECI}

\newpage
\begin{figure}[h!]\vspace*{+3cm}
\centering
	\sbox\mysavebox{
	\includegraphics[width=0.85\linewidth]{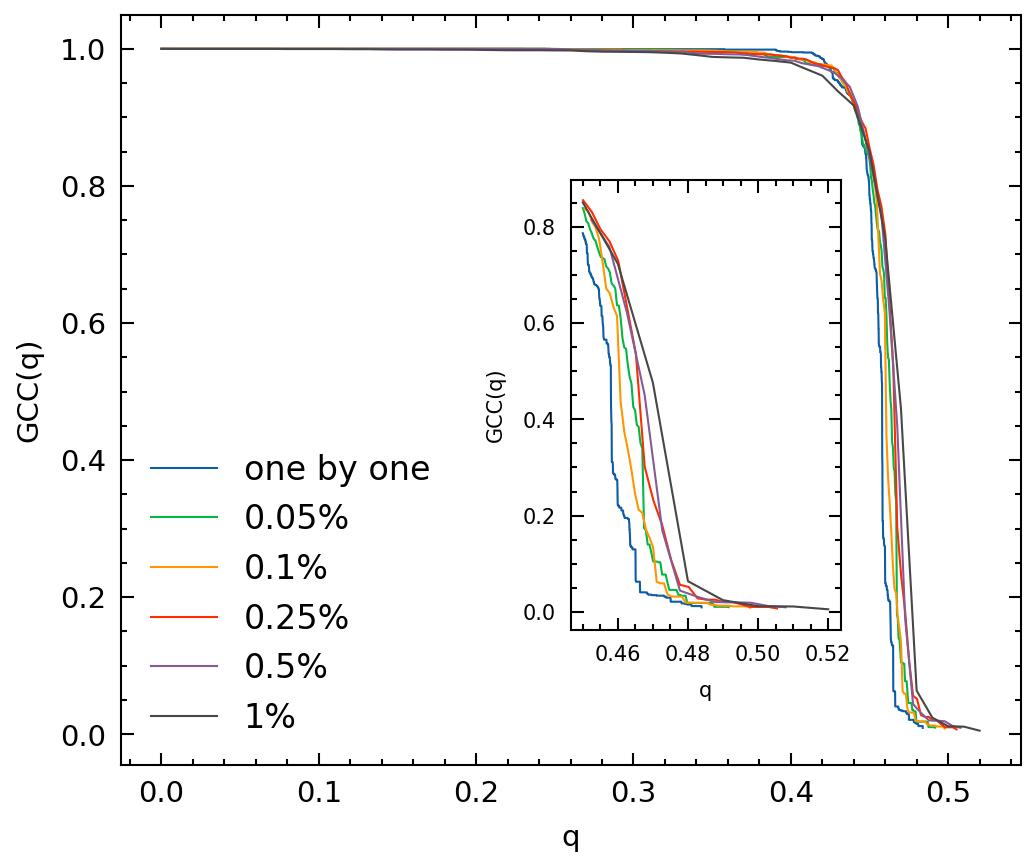}}%
	\usebox\mysavebox
  	\par
  	\begin{minipage}{\wd\mysavebox}\vspace*{+0.3cm}
		\caption{\small Comparison of different removal ratios for each step in an ER network with $N=10,000$ and $\langle k\rangle =3.5$. Inset: enlarged local results. } \label{fig:S1}
	\end{minipage}
\end{figure}

\newpage
\begin{figure}[h!]\vspace*{+3cm}
	\centering
	\sbox\mysavebox{
	\includegraphics[width=0.85\linewidth]{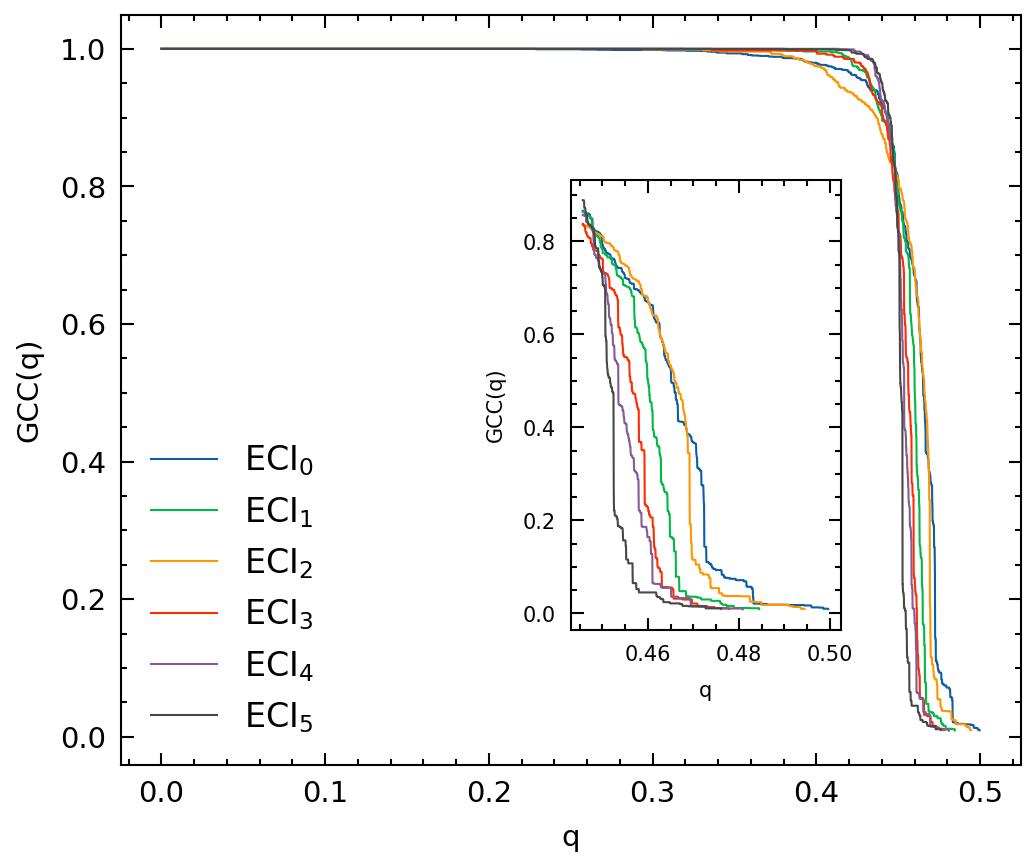}}%
	\usebox\mysavebox
  	\par
  	\begin{minipage}{\wd\mysavebox}\vspace*{+0.3cm}
		\caption{\small Comparison of the ECI algorithm for different radii $\ell$ (=0, 1, 2, 3, 4, 5) in an ER network with $N=10,000$ and $\langle k\rangle =3.5$. Inset: enlarged local results.}\label{fig:S2}
	\end{minipage}
\end{figure}

\newpage
\begin{figure}[h!]\vspace*{+3cm}
	\centering
	\sbox\mysavebox{
	\includegraphics[width=0.85\linewidth]{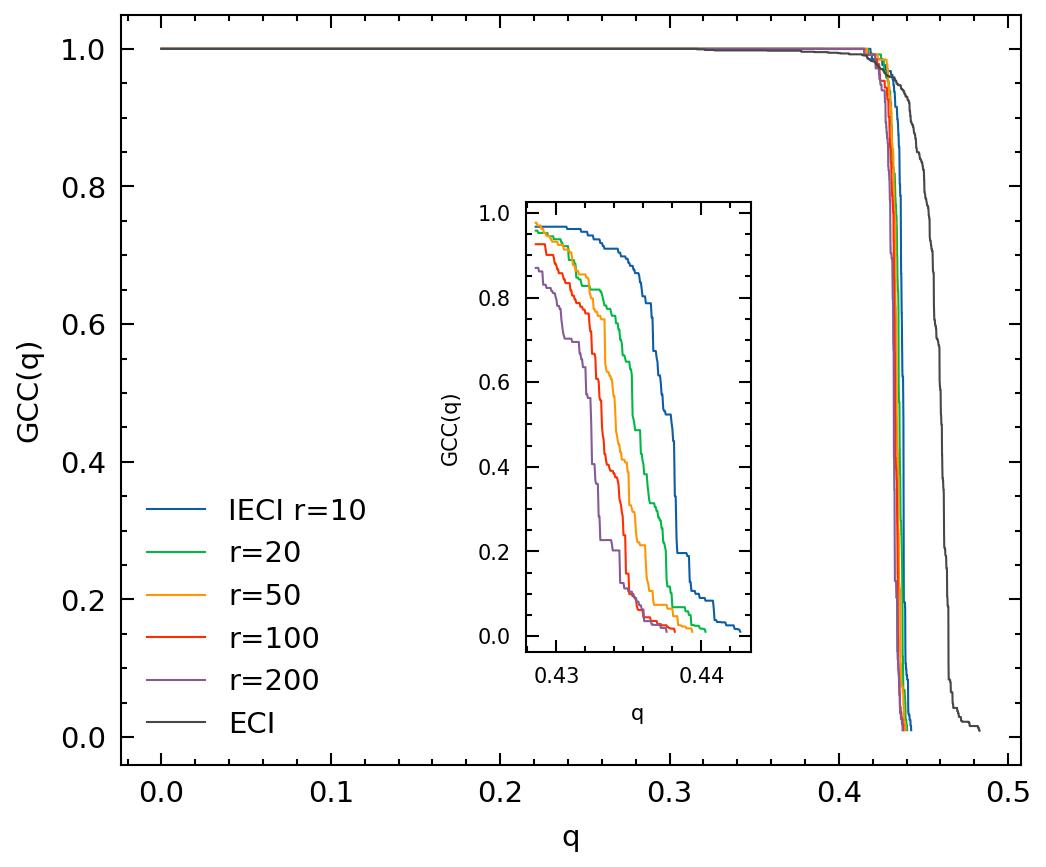}}%
	\usebox\mysavebox
  	\par
  	\begin{minipage}{\wd\mysavebox}\vspace*{+0.3cm}
		\caption{Improvement of the IECI algorithm with different $r$. We use $r=10, 20, 50, 100, 200$ in an ER network with $N=10,000$ and $\langle k\rangle =3.5$. Inset: enlarged local results.}\label{fig:S3}
	\end{minipage}
\end{figure}

\newpage
\begin{figure}[h!]\vspace*{+3cm}
	\centering
	\sbox\mysavebox{
	\includegraphics[width=0.85\linewidth]{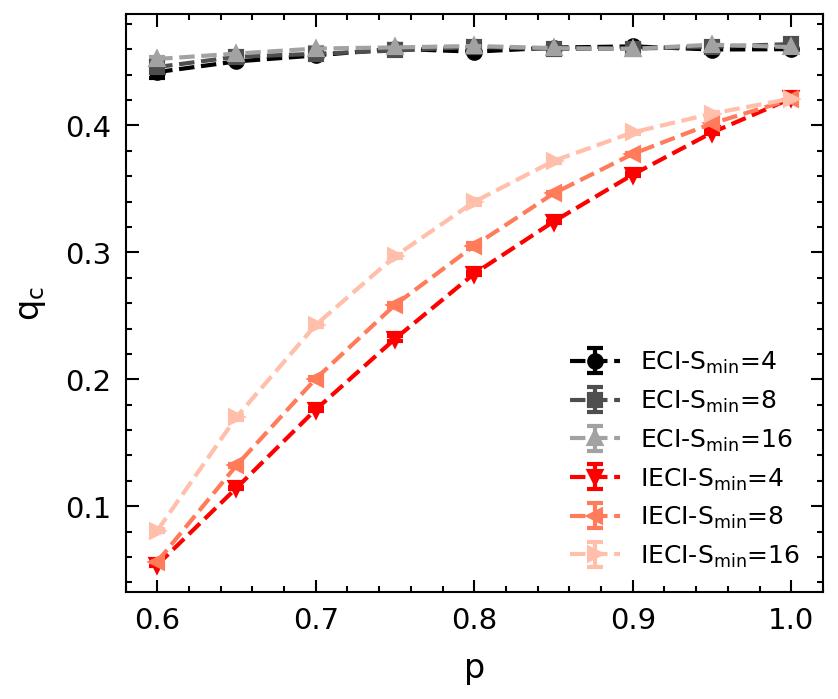}}%
	\usebox\mysavebox
  	\par
  	\begin{minipage}{\wd\mysavebox}\vspace*{+0.3cm}
		\caption{The dismantling thresholds of ECI and IECI versus $p$ in IDCP networks with parameters $S_{\min}=4, 8, 16$ (error bars are s.e.m. over 20 realizations).
		}\label{fig:S4}
	\end{minipage}
\end{figure}

\newpage
\begin{figure}[h!]\vspace*{+3cm}
  \centering
  \includegraphics[width=0.8\textwidth]{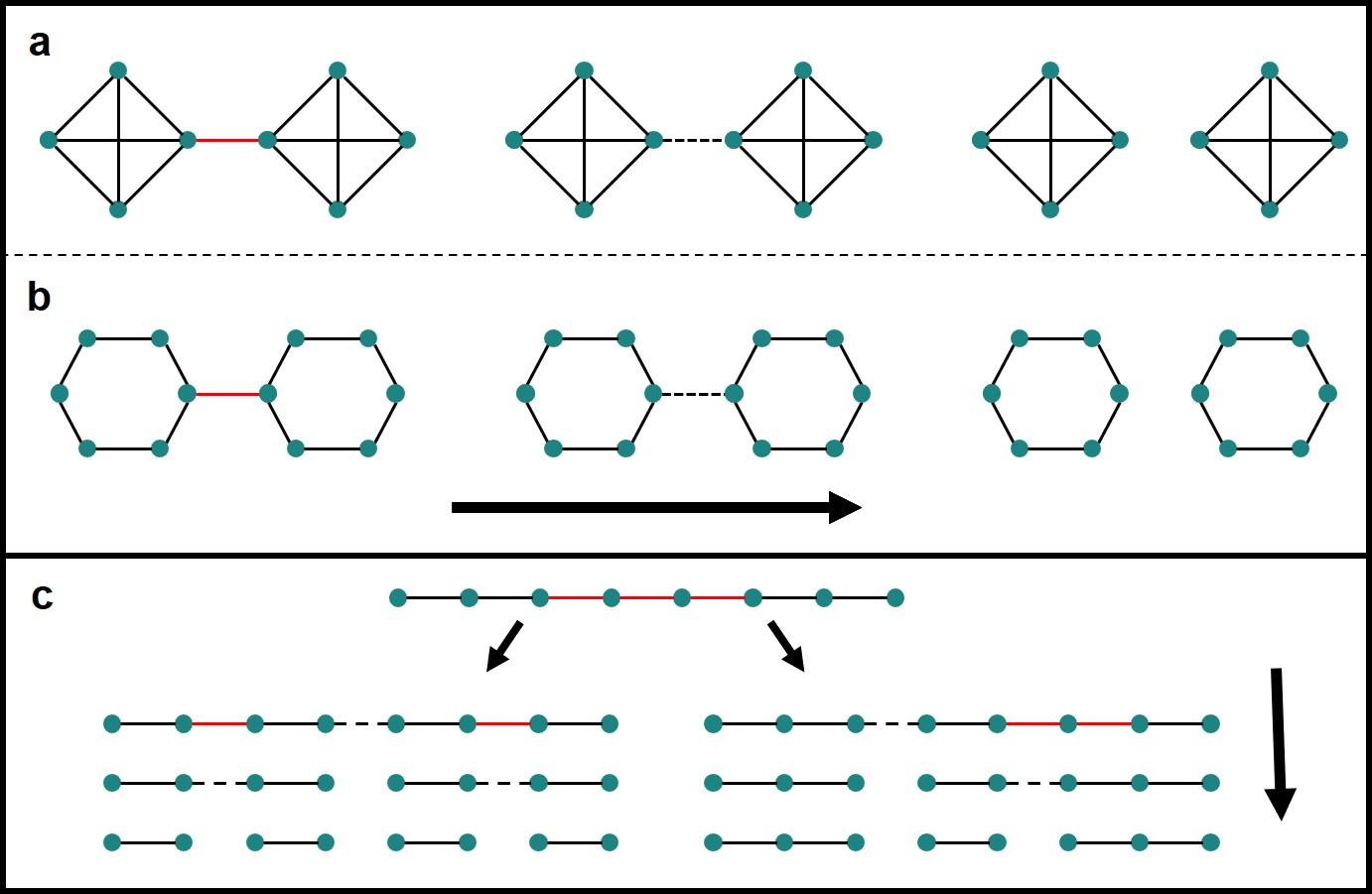}
  \caption{The procedure of ECI in some particular networks. The red edges indicate the edges with the greatest collective influence, and the dashed edges indicate the removed edges in the previous step.}
  \label{fig:S5}
\end{figure}
%%%%%%%%%%%%%%%%%%%%%%%%%%%%%%%
%TC:endignore
\end{document}